\documentclass[a4j,12pt]{article}
\usepackage{amsmath}
\usepackage{amssymb}
\usepackage{mathrsfs}
\usepackage{latexsym}
\usepackage{enumerate}
\usepackage{ntheorem}
\usepackage{graphicx}
\usepackage{here}
\usepackage{amscd}
\usepackage{url}
\usepackage{multirow}
\allowdisplaybreaks[4]

\theoremstyle{break}
\theoremheaderfont{\normalfont\bfseries}
\theoremseparator{.}
\newtheorem{definition}{Definition}[section]
\newtheorem{theorem}[definition]{Theorem}

\newtheorem{lemma}[definition]{Lemma}
\newtheorem{corollary}[definition]{Corollary}
\newtheorem*{proof}{Proof}[section]

\makeatletter
\providecommand{\bigsqcap}{%
  \mathop{%
    \mathpalette\@updown\bigsqcup
  }%
}
\newcommand*{\@updown}[2]{%
  \rotatebox[origin=c]{180}{$\m@th#1#2$}%
}
\makeatother

\newcommand{\jump}[1]{\ensuremath{[\![#1]\!]} }
\def\PP{\mathrm{PP}}
\def\TP{\mathrm{TP}}
\def\T{\mathrm{T}}

\title{An Intuitionistic Set-theoretical Model of the Extended Calculus of Constructions}
\author{Masahiro Sato}
\date{}

\begin{document}
\maketitle




\begin{abstract}
  Werner's set-theoretical model is one of the most intuitive models of ECC.
  It combines a functional view of predicative universes with a collapsed view of the impredicative sort $\mathrm{Prop}$.
  However this model of $\mathrm{Prop}$ is so coarse that the principle of excluded middle $P \lor \neg P$ holds.
  In this paper, we interpret $\mathrm{Prop}$ into a topological space (a special case of Heyting algebra) to make it more intuitionistic without sacrificing simplicity.
  We prove soundness and show some applications of our model.
\end{abstract}

\section{Introduction}\label{introduction}
There are various models of type systems.
Werner's Set-theoretical model~\cite{SetsInTypes} provides an intuitive model of ECC.
It combines a functional view of predicative universes with a collapsed view of the impredicative sort Prop.
However this model of Prop is so coarse that the principle of excluded middle $P \lor \neg P$ holds in it.

In this paper, we construct a set-theoretical model of ECC in which the principle of excluded middle $P \lor \neg P$ doesn't hold, and thus closer to completeness.

ECC (the Extended Calculus of Constructions) extends CC with a hierarchy of predicative sorts $\mathrm{Type}_i$ and strong sums $\Sigma x : A.B$.
CC (the Calculus of Constructions~\cite{CC}) is a pure type system~\cite{pure_type_system} with two sorts, impredicative $\mathrm{Prop}$ and predicative $\mathrm{Type}$.

In~\cite{SetsInTypes}, Werner provides a remarkably simple model of ECC without strong sums.
In this model, $\lambda x : A.t$ is interpreted by a set-theoretical function for predicative sorts.
Yet such a simple approach is known to fail for impredicative sorts as it runs afoul of Reynolds' paradox~\cite{model_not_set}.
Therefore, the model for $\mathrm{Prop}$ is two-valued. 
Hence the principle of excluded middle $P \lor \neg P$ is valid in this model, making it classical.
This simple approach is to be contrasted with Luo's original model of ECC which uses $\omega$-sets~\cite{Luo}, or more recent models such as categorical models~\cite{CategoryType} or models based on homotopy theory~\cite{HoTT}.
This is the drawback of simplicity: while this approach avoids many complications of more precise models, it is at times counter-intuitive, as it completely ignores the intuitionistic aspect of CC.
Our goal has been to recover the intuitionistic part of CC without increasing the complexity of the model.
To do this, we interpret $\mathrm{Prop}$ into some topological space.
Topological spaces are instances of Heyting algebras.
Heyting algebras are used when constructing models of intuitionistic logic, but usually their elements are not sets.
In our model, proofs shall be interpreted as elements of denotations of propositions, hence these denotations must be sets.
Using topological spaces solves this problem.
Despite the fact that the interpretation of $\mathrm{Prop}$ is many valued, we avoid Reynolds' paradox by making the interpretation of proofs undistinguished.
Due to proof-irrelevance, this model still validates some propositions that are not provable, hence this model does not reach completeness yet.
However this is sufficient to exclude many classical propositions such as the principle of excluded middle $P \lor \neg P$.
Note that, to make the model coherent, we had to slightly restrict the type system.
We believe the scope is still sufficient to make this model practical, but hope to remove these restrictions in the future.

This model is parametrized by a topological space $(X, \mathcal{O}(X))$ and a point $p \in X$, which is called the {\em reference point}\footnote{
  Our proof of soundness requires this reference point to satisfy a condition, which is called the {\em point condition}.
}.
By replacing the parameters of the model, we can make it more or less precise.
For instance if its parameters are the topological space $(\{\cdot\}, \{\phi, \{\cdot\}\})$ and the reference point `$\cdot$',
we obtain a model of classical logic, which is the coarsest one.
It suffices to add one more point and shift the reference point to invalidate the principle of excluded middle.

In section \ref{typetheory}, we define the language of the type system ECC.
In section \ref{model}, we give our set-theoretical interpretation of ECC, and prove its soundness.
In section \ref{application}, we show some applications of this model.
For instance, we show that the excluded middle cannot be derived from the linearity axiom in ECC.
In section \ref{reynolds}, we analyze how we avoid Reynolds' paradox.

\section{Definition of ECC}\label{typetheory}
We define the type system $ECC$ as follows (omitting strong sums, as in~\cite{SetsInTypes}).
\begin{definition}[Term]\label{definition_of_term}
  \
  \begin{itemize}
  \item $x$ is a term for $x \in V$.
  \item If $t_1$ and $t_2$ are terms, then $t_1 t_2$ is a term.
  \item If $t$ and $T$ is are terms, and $x \in V$ then, $\lambda x : T. t$ is a term.
  \item If $T_1$ and $T_2$ are terms, and $x \in V$ then $\forall x : T_1. T_2$ is a term.
  \item $\mathrm{Prop}, \mathrm{Type}_i$ are terms ($i = 0,1,2,3,4,...$).
  \end{itemize}
\end{definition}
$\mathrm{Prop}$ and $\mathrm{Type}_i$ are called {\em sorts}.
$\mathrm{Prop}$ is called the impredicative sort and it represents the type of all propositions.
$\mathrm{Type}_0$ is named ``$\mathrm{Set}$'' in Coq.

\begin{definition}[Context]
  \ 
  \begin{itemize}
  \item $[]$ is a context.
  \item If $\Gamma$ is a context, and $T$ is a term and $x \in V$, then $\Gamma ; (x : T)$ is a context.
  \end{itemize}
\end{definition}

We show the typing rules of ECC in Table \ref{tab:typing_rule_of_ECC}.
They are standard, except that we restricted the PI-Type rule in the case $P : \mathrm{Prop}$ and $Q : \mathrm{Prop}$, and removed the subtyping rule from Prop to Type.
The unrestricted Prop-Prop PI-Type rule creates difficulties when building an intuitionistic model, and if we do not remove the subtyping rule it becomes possible to use the Prop-Type case of the PI-Type rule in place of the restricted Prop-Prop case, which would make the model incoherent.
We believe these restrictions are reasonable, as the proof component is seldom used in the PI-Type rule, with the notable exception of the generic statement of proof-irrelevance. Removing the subtyping between Prop and Type does not change the expressive power, as it is still possible to explicitly duplicate properties using Type to Prop.
We hope to solve these problems in the future, and allow the standard typing rules.

\begin{table}[h]
  \centering
  {\scriptsize
    {\renewcommand\arraystretch{3.0}
      \begin{tabular}{|ccc|} \hline
        $\Gamma \vdash  \mathrm{Prop} : \mathrm{Type}_i$ &
        $\Gamma \vdash  \mathrm{Type}_i : \mathrm{Type}_{i+1}$ &
        (Axiom)\\
        $\cfrac{\Gamma \vdash  A : \mathrm{Type}_i}{\Gamma \vdash  A : \mathrm{Type}_{i+1}}$ &&
        (Subtyping) \\
        $\cfrac{\Gamma \vdash  A : \mathrm{Type}_i \quad \Gamma ; (x : A) \vdash  B : \mathrm{Type}_j}{\Gamma \vdash  \forall x:A.B:\mathrm{Type}_{\mathrm{max}(i,j)}}$ &
        $\cfrac{\Gamma \vdash  A : \mathrm{Prop} \quad \Gamma ; (x : A) \vdash  B : \mathrm{Type}_j}{\Gamma \vdash  \forall x:A.B.\mathrm{Type}_j}$ 
        & \multirow{2}{*}{(PI-Type)}\\
        
        $\cfrac{\Gamma \vdash  A : \mathrm{Type}_i \quad \Gamma;(x : A) \vdash  Q : \mathrm{Prop}}{\Gamma \vdash  \forall x:A.Q:\mathrm{Prop}}$ &  
        \multicolumn{2}{c|}{$\cfrac{\Gamma \vdash  P : \mathrm{Prop} \quad \Gamma \vdash  Q:\mathrm{Prop} \quad x \mbox{ does not appear in } Q}{\Gamma \vdash  \forall x:P.Q : \mathrm{Prop}}$}
        \\
        $\cfrac{\Gamma ; (x : A) \vdash  t : B \quad \Gamma \vdash  \forall x:A.B: \mathrm{Type_i}}{\Gamma \vdash  \lambda x:A.t : \forall x:A.B}$ &
        $\cfrac{\Gamma ; (x : A) \vdash t : B \quad \Gamma \vdash \forall x:A.B: \mathrm{Prop}}{\Gamma \vdash \lambda x:A.t : \forall x:A.B}$ &
        (Abstract) \\
        $\cfrac{\Gamma \vdash  u : \forall x:A.B \quad \Gamma \vdash  v : A \quad}{\Gamma \vdash  (u v) : B[x \backslash v]}$  &
        &
        (Apply) \\
        $\cfrac{(x : A) \in \Gamma \quad \Gamma \vdash A : \mathrm{Type}_i}{\Gamma \vdash x:A}$ &
        $\cfrac{(x : A) \in \Gamma \quad \Gamma \vdash A : \mathrm{Prop}}{\Gamma \vdash x:A}$ &
        (Variable) \\
        $\cfrac{\Gamma \vdash x : A \quad A =_{\beta} B}{\Gamma \vdash x : B}$ &
        &
        (Beta Equality) \\
        
        \hline
      \end{tabular}
    }
  }
  \caption{Typing Rule of ECC}
  \label{tab:typing_rule_of_ECC}
\end{table}

In Table \ref{tab:typing_rule_of_ECC}, $=_{\beta}$ denotes {\em beta equality} and $B[x \backslash v]$ denotes substitution.
They are defined in Definitions \ref{substitution_ecc} and \ref{beta_ecc} below.

\begin{definition}[Substitution]\label{substitution_ecc}
  Let $t$ and $v$ be terms and $x$ be a variable.
  The substitution $t[x \backslash v]$, which means $v$ replaces $x$ in $t$, is defined inductively as follows:
  \begin{enumerate}[{(}i{)}]
  \item If $y$ is a variable, then $y[x \backslash v] = \begin{cases}v & (y=x) \\ x & (otherwise),\end{cases}$
  \item $(t_1 t_2)[x \backslash v] = (t_1[x \backslash v]) (t_2[x \backslash v])$,
  \item $(\lambda x' : T.t')[x \backslash v] = \lambda x' : (T[x \backslash v]). t'[x \backslash v] \quad (\mbox{when $x \neq x'$})$,
  \item $(\forall x' : T_1 . T_2)[x \backslash v] = \forall x' : (T_1[x \backslash v]) . (T_2[x \backslash v])$,
  \item $(\mathrm{Prop})[x \backslash v] = \mathrm{Prop}$,
  \item $(\mathrm{Type}_i)[x \backslash v] = \mathrm{Type}_i \quad (i=1,2,3,...)$.
  \end{enumerate}
\end{definition}

\begin{definition}[Beta Equality]\label{beta_ecc}
  Let $=_{\beta}$ be the smallest equivalence relation such that following conditions hold.
  \begin{enumerate}[{(}i{)}]
  \item $(\lambda x : A. t) \; a =_{\beta} t[x \backslash a]$.
  \item If $t_1 =_{\beta} t_1'$ and $t_2 =_{\beta} t_2'$ then $t_1 t_2 =_{\beta} t_1' t_2'$.
  \item If $t =_{\beta} t'$ and $A =_{\beta} A'$ then $\lambda x : A. t =_{\beta} \lambda x : A' t'$.
  \item If $A =_{\beta} A'$ and $B =_{\beta} B'$ then $\forall x : A. B =_{\beta} \forall x : A' B'$.
  \end{enumerate}
\end{definition}

In ECC, propositions are types which belong to the impredicative sort $\mathrm{Prop}$,
and proofs are terms of types which represent propositions.
Next, we give a definition of proposition and proof as follows.
\begin{definition}
  \
  \begin{enumerate}
  \item Propositional Term \\
    The term $P$ is called a propositional term for $\Gamma$ iff $\Gamma \vdash P : \mathrm{Prop}$ is derivable.
  \item Proof Term \\
    The term $p$ is called a proof term for $\Gamma$ iff $\Gamma \vdash p : P$ is derivable for some $P$ which is a propositional term for $\Gamma$.
  \item Provable Propositional Term \\
    The term $P$ is called a provable propositional term for $\Gamma$ iff $P$ is a propositional term for $\Gamma$ and there exists $p$ such that $\Gamma \vdash p : P$ is derivable.
  \end{enumerate}
\end{definition}

Proof terms and propositional terms are preserved under substitution.
The following lemma expresses this fact.


\begin{lemma}\label{substitution_in_proof}
  The following statements are equivalent.
  \begin{itemize}
  \item $p$ is a proof(resp. propositional) term for the context $\Gamma ; (x : U) ; \Delta$.
  \item $p[x \backslash u]$ is a proof(resp. propositional) term for the context $\Gamma; \Delta[x \backslash u]$.
  \end{itemize}
\end{lemma}

This lemma is consequence of the following proposition.
\begin{lemma}\label{substitution_lemma}
  If $\Gamma \vdash u : U$ is derivable, then
  $\Gamma ; (x : U) ; \Delta \vdash t : T$ is derivable if and only if $\Gamma ; \Delta[x \backslash u] \vdash t[x \backslash u] : T[x \backslash u]$ is derivable.
\end{lemma}
Lemma \ref{substitution_lemma} can be proved in the same way as in~\cite{not_simple}.

Lastly, here are some notations allowing to use other logical symbols~\cite{lambda_type}.
\begin{definition}\label{logical_symbol}
  \begin{eqnarray*}
    A \rightarrow B &:=& \forall x : A. B \quad \text{(when `$x$' does not occur freely in `$B$')}, \\
    \bot &:=& \forall P : \mathrm{Prop}. P, \\
    \neg A &:=& A \rightarrow \bot, \\
    A \land B &:=& \forall P : \mathrm{Prop}. (A \rightarrow B \rightarrow P) \rightarrow P, \\
    A \lor B &:=& \forall P : \mathrm{Prop}. (A \rightarrow P) \rightarrow (B \rightarrow P) \rightarrow P, \\
    \exists x:A.Q &:=& \forall P : \mathrm{Prop}. (\forall x : A. (Q \rightarrow P)) \rightarrow P, \\
    A \leftrightarrow B &:=& (A \rightarrow B) \land (B \rightarrow A), \\
    x =_A y &:=& \forall Q : (A \rightarrow \mathrm{Prop}). Q \; x \leftrightarrow Q \; y.
  \end{eqnarray*}
\end{definition}

\section{Interpretation}\label{model}
\subsection{Lattice}
Several interpretations of type theory have been proposed such as using $\omega$-sets~\cite{Luo} or coherent spaces~\cite{coherent}.
In this paper, we use {\em Heyting algebras}~\cite{MacLane, IntuitionisticLogic} for propositions.
Heyting algebras provide models of intuitionistic logic.
Topological spaces form Heyting algebras, and as such provide models of intuitionistic logic too~\cite{IntuitionisticLogic}.
We give a definition of lattice and Heyting algebra as follows.
\begin{definition}[Lattice]
  Let $(A,\le)$ be a partial order set(i.e. reflexivity, antisymmetry, and transitivity).
  $(A,\le)$ is called Lattice when any two elements $a$ and $b$ of $A$ have a supremum `$a \sqcup b$' and infimum `$a \sqcap b$', which are called join and meet\footnote{
    We use the lattice operation symbols join `$\sqcup$' and meet `$\sqcap$' instead of `$\lor$' and `$\land$', since we use these in another way in this paper.
  }.
  A lattice is also called $complete$ $lattice$ if every subset $S$ of $A$ has supremum `$\bigsqcup S$' and infimun `$\bigsqcap S$'.
  If a lattice has an $exponential$ $operator$ $a^b$ such that
  \begin{equation*}
    x \le z^y \Leftrightarrow x \sqcap y \le z
  \end{equation*}
  holds, then we call it Heyting Algebra.
\end{definition}
The following lemma show that complete lattice is stronger than Heyting algebra.
\begin{lemma}
  If $(A, \le)$ is a complete lattice, then this is also a Heyting algebra.
\end{lemma}
\begin{proof}
  Since complete heyting algebra, we can define the exponential operator.
  \begin{equation*}
    y^x := \bigsqcup \{t \; | \; t \sqcap x \le y\}.
  \end{equation*}
\end{proof}

\begin{lemma}
  For any set $X$, the topological space $(X, \mathcal{O}(X))$ is a Heyting algebra, moreover it is a complete lattice.
\end{lemma}
\begin{proof}
  In fact let $a \le b$ be $a \subset b$, and define each operation as follows:
  \begin{eqnarray*}
    \mathbb{I} &:=& X, \\
    \mathbb{O} &:=& \phi, \\
    \bigsqcup S &:=& \bigcup S, \\
    \bigsqcap S &:=& \bigsqcup\{t \mid \forall s \in S, t \leq s\} = \biggl(\bigcap S \biggr)^\circ \quad (where \; A^\circ \; is \; the \; interior \; of \; A), \\
    b^a &:=& \bigsqcup\{t \mid t \sqcap a \leq b\}.
  \end{eqnarray*}
\end{proof}

The following lemma states well known properties of complete Heyting algebras.
\begin{lemma}\label{heyting_conditions}
  Let $(A,\leq)$ be a complete Heyting algebra. Then the following conditions hold.
  \begin{eqnarray}
    (x^b)^a &=& x^{a \sqcap b}, \label{eq:powerprod}\\
    \bigsqcap \{t^{t^a} \; | \; t \in A\} &=& a, \label{eq:meetpower}\\
    x^a \sqcap x^b &=& x^{a \sqcup b}, \label{eq:prodpoweror}\\
    \bigsqcap \{a^t \; | \; t \in S\} &=& a^{\bigsqcup S}, \label{eq:meetpoweror}\\
    \bigsqcap \phi &=& 1, \label{eq:whole}\\
    x &\leq& x^y, \label{eq:powerle}\\
    x^y \sqcap y ^ x = 1 &\Rightarrow& x = y, \label{eq:power_eqcond}\\
    \bigsqcap S = 1 &\Rightarrow& \forall a \in S, a = 1. \label{eq:topcond}
  \end{eqnarray}
\end{lemma}

\subsection{Preparation of interpretation}
Let $p$, which is called the {\em reference point}, be some point of the topological space $(X,\mathcal{O}(X))$ such that the following condition
\begin{equation*}
  \bigcap \mathcal{U}(p) \; is \; an \; open \; set
\end{equation*}
hold where $\mathcal{U}(p)$ is an open neighborhood\footnote{An open neighborhood of $p$ is a set of open sets containing the point $p$} of p.
We will parametrize our model with $\mathcal{O}(X)$ and $p$.
Let us call this condition the {\em point condition}. It becomes necessary when proving soundness.

\begin{definition}[Dependent Function]
  Let $A$ be a set, and $B(a)$ be a set with parameter $a \in A$.
  We define dependent function domain as follows
   \begin{equation*}
     \prod_{a \in A}B(a) := \{f \subset \coprod_{a \in A}B(a) \; | \; \forall a \in A, \exists ! b \in B(a), (a, b) \in f\}
   \end{equation*}
   that is functions whose graph belongs to
   \begin{equation*}
     \coprod_{a \in A} B(a) := \{(x,y) \in A \times \bigcup_{a \in A}B(a) \; | \; y \in B(x) \}.
  \end{equation*}
\end{definition}

The function $PT$ called {\em Product Type} is defined as follow.
\begin{definition}[Product Type]
  \begin{equation*}
    PT_{\Gamma, x}(A,B) := 
    \begin{cases}
      \PP & (A \mbox{ is a propositional term for }  \Gamma \; \\
      & \quad \mbox{ and } B \mbox{ is a propositional term for } \Gamma \\
      \TP & (A \mbox{ is not a propositional term for }  \Gamma \; \\
      & \quad \mbox{ and } B \mbox{ is a propositional term for } (\Gamma ; x : A))\\
      \T & (\mbox{otherwise}) \\
    \end{cases}
  \end{equation*}
\end{definition}

The function $PT_{\Gamma,x}$ maps two types into string symbols $\{\PP, \TP, \T\}$.
Its goal is to discriminate cases of $\forall x : A.B$ to give them different interpretations.

Next, we introduce the Grothendieck universes as in~\cite{SetsInTypes}.
\begin{definition}
  Let $\alpha$ be an ordinal.
  We define $V_\alpha$ as follows.
  \begin{itemize}
  \item $V_0 = \phi$
  \item $V_{\alpha} = \displaystyle\bigcup_{\beta < \alpha} \mathcal{P}(V_{\beta})$
  \end{itemize}
  And we define the Grothendieck universe $\mathscr{U}(i)$ as follows
  \begin{equation*}
    \mathscr{U}(i) = V_{\lambda_i}
  \end{equation*}
  where $\lambda_i$ is $i$-th inaccessible cardinal.
\end{definition}
The following lemma is necessary when proving soundness.
\begin{lemma}\label{prod_univ}
  \
  \begin{itemize}
  \item $A \in \mathscr{U}(i)$ and $B(a) \in \mathscr{U}(i)$ for each $a \in A$ imply $\displaystyle\prod_{a \in A}B(a) \in \mathscr{U}(i)$.
  \item $A \in \mathscr{U}(i)$ implies $A \subset \mathscr{U}(i)$.
  \end{itemize}
\end{lemma}

\subsection{Interpretation of the judgments}
In this model, a type $T$ is interpreted into a set $\jump{T}$, and a context $x_1 : T_1 ; x_2 : T_2 ; \cdots ; x_n : T_n$  is interpreted into a tuple in $\jump{T_1} \times \jump{T_2} \times \cdots \times \jump{T_n}$ when there are no dependent types in the context.

First, we define the interpretation of contexts $\jump{\mathrm{-}}$, judgments $\jump{\mathrm{-} \vdash \mathrm{-}}$ and strict judgments $\jump{\mathrm{-} \vdash \mathrm{-}}'$ by mutual recursion as follows.

\begin{definition}[interpretation]\label{interpretation}
  Let $(X, \mathcal{O}(X))$ be a topological space, and $p$ be a reference point of $X$ satisfying the {\em point condition}.
  \begin{enumerate}[{(}i{)}]
  \item Definition of the strict-interpretation of a judgment $\jump{\mathrm{-} \vdash \mathrm{-}}'$ \\
    \begin{equation*}
      \jump{\Gamma \vdash A}' (\gamma) = \begin{cases} \jump{\Gamma \vdash A}(\gamma) \cap \{p\} & (A \mbox{ is a propositional term in } \Gamma) \\
        \jump{\Gamma \vdash A} (\gamma) & (\mbox{otherwise}) \end{cases}
    \end{equation*}
  \item Definition of the interpretation of a context $\jump{\mathrm{-}}$ \\
    \begin{eqnarray*}
      \jump{[]} &:=& \{()\} \\
      \jump{\Gamma ; (x:A)} &:=&
      \{(\gamma, \alpha) \mid \gamma \in \jump{\Gamma}  \; \mathrm{and} \; \alpha \in \jump{\Gamma \vdash A}'(\gamma) \} \\
      &=& \coprod_{\gamma \in \jump{\Gamma}} \jump{\Gamma \vdash A}'(\gamma)
    \end{eqnarray*}
    
  \item Definition of the interpretation of a judgment $\jump{\mathrm{-} \vdash \mathrm{-}}$ \\
    If $t$ is a proof term for $\Gamma$, then
    \begin{equation*}
      \jump{\Gamma \vdash t} = p
    \end{equation*}
    otherwise,
    \begin{eqnarray*}
      \jump{\Gamma \vdash \mathrm{Type}_i}(\gamma) &:=& \mathscr{U}(i) \\
      \jump{\Gamma \vdash \mathrm{Prop}}(\gamma) &:=& \mathcal{O}(X) \\
      \jump{\Gamma \vdash \forall x:P.Q}(\gamma) &:=&
      \begin{cases}
        { \biggl(\jump{\Gamma \vdash Q}(\gamma) \biggr) }^{\jump{\Gamma \vdash P}(\gamma) } \\
        \quad\quad (\mbox{when} \; PT_{\Gamma,x}(P,Q) = \PP) \\
        & \\
        \bigsqcap \{ \jump{\Gamma ; (x:P) \vdash Q}(\gamma, \alpha) \mid \alpha \in \jump{\Gamma \vdash P}(\gamma) \} \\
        \quad\quad (\mbox{when} \; PT_{\Gamma, x}(P,Q) = \TP) \\
        & \\
        \prod_{\alpha \in \jump{\Gamma \vdash P}'(\gamma)} \jump{\Gamma ; (x:P) \vdash Q}(\gamma, \alpha) \\
        \quad\quad (\mbox{when} \; PT_{\Gamma, x}(P,Q) = \T)
      \end{cases} \\
      \jump{\Gamma \vdash \lambda x:A.t}(\gamma) &:=& \Bigl\{\bigl(\alpha, \jump{\Gamma ; (x : A) \vdash t}(\gamma, \alpha)\bigr) \; | \; \alpha \in \jump{\Gamma \vdash A}'(\gamma) \Bigr\} \\
      \jump{\Gamma \vdash u v}(\gamma) &:=&  \jump{\Gamma \vdash u}(\gamma)\biggl(\jump{\Gamma \vdash v}(\gamma)\biggr) \\
      \jump{\Gamma \vdash x_i}(\gamma) &:=& \gamma_i
    \end{eqnarray*}
  \end{enumerate}
  For simplicity, we write $\jump{T}$ for $\jump{[]\vdash T}()$, when the context is empty.
\end{definition}

The interpretation of a context $\jump{\Gamma}$ is a sequence whose length is the length of $\Gamma$. $\jump{\Gamma \vdash t}$ is the function whose domain is $\Gamma$ and which maps to some set.
Most cases are similar to Werner's interpretation, so we only explain the interpretation of $\forall x : P.Q$.
There are three cases, according to the result of $PT_{\Gamma, x}(P,Q)$.
When $PT_{\Gamma, x}(P,Q) = \PP$, the interpretation of $\jump{\Gamma \vdash \forall x : P.Q}$ represents the logical implication $P \Rightarrow Q$.
We use the Heyting algebra representation of this implication.
Here we assume that $x$ does not appear in $Q$, thanks to our restriction.
Otherwise we would need to build the interpretation of $\jump{\Gamma ; (x : P) \vdash Q}(\gamma, p)$, but this requires that $p \in \jump{\Gamma \vdash P}(\gamma)$, which is not always true.
When $PT_{\Gamma, x}(P,Q) = \TP$ the interpretation of $\jump{\Gamma \vdash \forall x : P.Q}$ represents universal quantification, and again we use the infinite meet operator of the complete Heyting algebra to express it.
In the last case only the representation becomes a set theoretical dependent function.

We start with the substitution lemma as follows.
Thus our interpretation is well behaved.
\begin{lemma}[substitution lemma]\label{substitution_interpretation}
  We assume $\Gamma \vdash u : U$ is derivable.
  If
  \begin{equation*}
    (\gamma, \jump{\Gamma \vdash u}(\gamma), \delta) \in \jump{\Gamma; (x : U) ; \Delta} 
  \end{equation*}
  holds for any $\gamma$ and $\delta$, then
  \begin{equation*}
    \jump{\Gamma ; \Delta[x \backslash u] \vdash t[x \backslash u]}(\gamma,\delta) = \jump{\Gamma; (x:U) ; \Delta \vdash t}(\gamma, \jump{\Gamma \vdash u}(\gamma),\delta)
  \end{equation*}
  for all $t$ and $\Delta$.
\end{lemma}
This lemma appears already in~\cite{SetsInTypes} and~\cite{not_simple}.
To prove it, we introduce the following two lemmas.
\begin{lemma}\label{interpretation_composition}
  \begin{equation*}\jump{\Gamma \vdash u}(\gamma) = \jump{\Gamma ; \Delta \vdash u}(\gamma, \delta)\end{equation*}
\end{lemma}

\begin{lemma}\label{interpretation_reduce}
  If $\jump{\Gamma ; (x : U) ; \Delta \vdash t}$ is well-defined, then so is $\jump{\Gamma ; \Delta[x \backslash u] \vdash t[x \backslash u]}$. And more, $(\gamma, \jump{\Gamma \vdash u}(\gamma), \delta) \in \jump{\Gamma ; (x : U) ; \Delta}$ implies $(\gamma, \delta) \in \jump{\Gamma; \Delta[x \backslash u]}$.
\end{lemma}
Next, we are ready to prove the substitution lemma \ref{substitution_interpretation}.
\begin{proof}[Proof of Lemma \ref{substitution_interpretation}]
  If $t$ is a proof term, it is clear by Lemma \ref{substitution_in_proof}.
  If $t$ is not a proof term, it is provable by induction on term $t$ by using Lemma \ref{interpretation_composition} and \ref{interpretation_reduce}.
\end{proof}

Finally we prove the following theorem about the interpretation of logical symbols in definition \ref{logical_symbol}.
It demonstrates the validity of the interpretation.
\begin{theorem}[interpretation of logical symbols]
  \
  \begin{enumerate}[{(}i{)}]
  \item $\jump{\Gamma \vdash \bot} = \phi$
  \item $\jump{\Gamma \vdash A \land B}(\gamma) = (\jump{\Gamma \vdash A}(\gamma)) \sqcap (\jump{\Gamma \vdash B}(\gamma))$
  \item $\jump{\Gamma \vdash A \lor B}(\gamma) = (\jump{\Gamma \vdash A}(\gamma)) \sqcup (\jump{\Gamma \vdash B}(\gamma))$
  \item $\jump{\Gamma \vdash \exists x : A. Q}(\gamma) = \displaystyle\bigsqcup_{\alpha \in \jump{\Gamma \vdash A}(\gamma)} \jump{\Gamma ; (x:A) \vdash Q}(\gamma, \alpha)$
  \item $\jump{\Gamma \vdash A \leftrightarrow B}(\gamma) = X \Rightarrow \jump{\Gamma \vdash A}(\gamma) = \jump{\Gamma \vdash B}(\gamma)$
  \item $\jump{\Gamma \vdash x =_A y}(\gamma) = X \Rightarrow \jump{\Gamma \vdash x}(\gamma) = \jump{\Gamma \vdash y}(\gamma)$
  \end{enumerate}
\end{theorem}
\begin{proof}
  Let $a, b, q(\alpha)$ be
  \begin{eqnarray*}
    a &:=& \jump{\Gamma \vdash A}(\gamma) \\
    b &:=& \jump{\Gamma \vdash B}(\gamma) \\
    q(\alpha) &:=& \jump{\Gamma ; (x : A) \vdash Q}(\gamma, \alpha).
  \end{eqnarray*} 
  By using Lemma \ref{heyting_conditions} and Lemma \ref{interpretation_composition} we have the followings:
  \begin{enumerate}[{(}i{)}]
  \item The proof of $\jump{\Gamma \vdash \bot} = \phi$.
    \begin{eqnarray*}
      \jump{\Gamma \vdash \bot}(\gamma) &=& \jump{\Gamma \vdash \forall P : \mathrm{Prop}. P}(\gamma) \\
      &=& \bigsqcap \{\jump{\Gamma ; (P : \mathrm{Prop}) \vdash P}(\gamma, x) \; | \; x \in \jump{\Gamma \vdash \mathrm{Prop}}(\gamma) \} \\
      &=& \bigsqcap \{x | x \in \mathcal{O}(X) \} \\
      &=& \phi
    \end{eqnarray*}
  \item The proof of $\jump{\Gamma \vdash A \land B}(\gamma) = (\jump{\Gamma \vdash A}(\gamma)) \sqcap (\jump{\Gamma \vdash B}(\gamma))$.
    \begin{eqnarray*}
      \jump{\Gamma \vdash A \land B}(\gamma) &=& \jump{\Gamma \vdash \forall P : \mathrm{Prop}. (A \rightarrow (B \rightarrow P)) \rightarrow P}(\gamma) \\
      &=& \bigsqcap \{ x^{(x^b)^a} \; | \; x \in \mathcal{O}(X) \} \\
      &=& \bigsqcap \{ x^{x^{a \sqcap b}} \; | \; x \in \mathcal{O}(X) \} \quad (\mathrm{by \; Lemma} \; \ref{heyting_conditions} \; (\ref{eq:powerprod}))\\
      &=& a \sqcap b \quad (\mathrm{by \; Lemma} \; \ref{heyting_conditions} \; (\ref{eq:meetpower}))\\
      &=& \jump{\Gamma \vdash A}(\gamma) \sqcap \jump{\Gamma \vdash B}(\gamma)
    \end{eqnarray*}
  \item The proof of $\jump{\Gamma \vdash A \lor B}(\gamma) = (\jump{\Gamma \vdash A}(\gamma)) \sqcup (\jump{\Gamma \vdash B}(\gamma))$.
    \begin{eqnarray*}
      \jump{\Gamma \vdash A \lor B}(\gamma) &=& \jump{\Gamma \vdash \forall P : \mathrm{Prop}. (A \rightarrow P) \rightarrow ((B \rightarrow P) \rightarrow P)}(\gamma) \\
      &=& \bigsqcap \{ (x^{x^b})^{x^a} \; | \; x \in \mathcal{O}(X) \} \\
      &=& \bigsqcap \{ x^{x^a \sqcap x^b} \; | \; x \in \mathcal{O}(X) \} \quad (\mathrm{by \; Lemma} \; \ref{heyting_conditions} \; (\ref{eq:powerprod}))\\
      &=& \bigsqcap \{ x^{x^{a \sqcup b}} \; | \; x \in \mathcal{O}(X) \} \quad (\mathrm{by \; Lemma} \; \ref{heyting_conditions} \; (\ref{eq:prodpoweror}))\\
      &=& a \sqcup b \quad (\mathrm{by \; Lemma} \; \ref{heyting_conditions} \; (\ref{eq:meetpower}))\\
      &=& \jump{\Gamma \vdash A}(\gamma) \sqcup \jump{\Gamma \vdash B}(\gamma)
    \end{eqnarray*}
  \item The proof of $\jump{\Gamma \vdash \exists x : A. Q}(\gamma) = \displaystyle\bigsqcup_{\alpha \in \jump{\Gamma \vdash A}(\gamma)} \jump{\Gamma ; (x:A) \vdash Q}(\gamma, \alpha)$.
    \begin{eqnarray*}
      \jump{\Gamma \vdash \exists a : A. Q}(\gamma) &=& \jump{\Gamma \vdash \forall P : \mathrm{Prop}. (\forall a : A. (Q \rightarrow P) \rightarrow P}(\gamma) \\
      &=& \bigsqcap \{ x^{\bigsqcap\{x^{q(\alpha)} \; | \; \alpha \in a\}} \; | \; x \in \mathcal{O}(X) \} \\
      &=& \bigsqcap \{ x^{x^{\bigsqcup\{q(\alpha) \; | \; \alpha \in a\}}} \; | \; x \in \mathcal{O}(X)\} \quad (\mathrm{by \; Lemma} \; \ref{heyting_conditions} \; (\ref{eq:meetpoweror}))\\
      &=& \bigsqcup\{q(\alpha) \; | \; \alpha \in a \} \quad (\mathrm{by \; Lemma} \; \ref{heyting_conditions} \; (\ref{eq:meetpower}))\\
      &=& \bigsqcup_{\alpha \in \jump{\Gamma \vdash A}(\gamma)}\jump{\Gamma ; (a : A) \vdash Q}(\gamma, \alpha)
    \end{eqnarray*}
  \item The proof of $\jump{\Gamma \vdash A \leftrightarrow B}(\gamma) = X \Rightarrow \jump{\Gamma \vdash A}(\gamma) = \jump{\Gamma \vdash B}(\gamma)$.
    \begin{eqnarray*}
      \jump{\Gamma \vdash A \leftrightarrow B}(\gamma) &=& \jump{\Gamma \vdash A \rightarrow B}(\gamma) \sqcap \jump{\Gamma \vdash B \rightarrow A}(\gamma) \\
      &=& a^b \sqcap b^a
    \end{eqnarray*}
    Hence we have $a=b$ by Lemma \ref{heyting_conditions} (\ref{eq:power_eqcond}) since $a^b \sqcap b^a = X$.
  \item The proof of $\jump{\Gamma \vdash x =_A y}(\gamma) = X \Rightarrow \jump{\Gamma \vdash x}(\gamma) = \jump{\Gamma \vdash y}(\gamma)$.
    \begin{eqnarray*}
      \jump{\Gamma \vdash x =_A y}(\gamma) &=& \jump{\Gamma \vdash \forall Q : (A \rightarrow \mathrm{Prop}). Q \; x \leftrightarrow Q \; y}(\gamma) \\
      &=& \bigsqcap_{f : a \rightarrow \mathcal{O}(X)}\jump{\Gamma ; (Q : A \rightarrow \mathrm{Prop}) \vdash Q \; x \leftrightarrow Q \; y}(\gamma, f)
    \end{eqnarray*}
    Since $\jump{\Gamma \vdash x =_A y}(\gamma) = X$ and Lemma \ref{heyting_conditions} (\ref{eq:topcond}), we have the following fact:
    \begin{equation*}
      \forall f : a \rightarrow \mathcal{O}(X), \jump{\Gamma ; (Q \rightarrow \mathrm{Prop}) \vdash Q \; x \leftrightarrow Q \; y}(\gamma, f) = X
    \end{equation*}
    Therefore we have $f(\jump{\Gamma \vdash x}(\gamma)) = f(\jump{\Gamma \vdash y}(\gamma))$ for any function $f : a \rightarrow \mathcal{O}(X)$.
    Hence, the statement holds.
  \end{enumerate}
\end{proof}

\subsection{Soundness}
We are ready to prove the soundness of this type system.

\begin{theorem}[soundness]\label{soundness}
  We assume $\jump{\Gamma}$ is non empty set.
  \begin{enumerate}
  \item If $t_1 =_\beta t_2$, and $\Gamma \vdash t_1 : T, \Gamma \vdash t_2 : T$ is derivable, then $\jump{\Gamma \vdash t_1}(\gamma) = \jump{\Gamma \vdash t_2}(\gamma)$.
  \item If $\Gamma \vdash t : T$ is derivable and $\jump{\Gamma}$ is non-empty set, then $\jump{\Gamma \vdash t}(\gamma) \in \jump{\Gamma \vdash T}(\gamma)$.
  \end{enumerate}
\end{theorem}

\begin{proof}[Proof of Theorem \ref{soundness}]
  \ \\
  1. It is sufficient that $\jump{\Gamma \vdash (\lambda x :U. t) \; u}(\gamma) = \jump{\Gamma \vdash t[x \backslash u]}(\gamma)$.
  By using Lemma \ref{substitution_interpretation}, 
  \begin{eqnarray*}
    && \jump{\Gamma \vdash (\lambda x : U.t) u} \\
    &=& \jump{\Gamma \vdash \lambda x : U.t}(\gamma)\bigl(\jump{\Gamma \vdash u}(\gamma)\bigr) \\
    &=& \jump{\Gamma ; (x : U) \vdash t}(\gamma, \jump{\Gamma \vdash u}(\gamma)) \\
    &=& \jump{\Gamma \vdash t[x \backslash u]}(\gamma)
  \end{eqnarray*}
  Hence, the statement holds. \\
  \\
  2. This is proved by induction on the Typing Rules in Table~\ref{tab:typing_rule_of_ECC}.
  For details, see Appendix~\ref{soundness_ind}.
  We must be careful in the case of Abstraction, i.e. $T = \forall x : A.B$ and $PT_{\Gamma, x}(A, B) = \TP$.
  To prove the soundness, we need the following equation
  \begin{equation*}
    \jump{\Gamma \vdash \forall x : A.B}(\gamma) = \bigcap\{\jump{\Gamma ; (x : A) \vdash B}(\gamma, \alpha) | \alpha \in \jump{\Gamma \vdash A}(\gamma) \}.
  \end{equation*}
  This equation does not hold in general, however we can obtain it by assuming the point condition at $p$.
\end{proof}

\begin{corollary}
  If $P$ is a provable propositional term for $\Gamma$, then 
  \begin{equation*}
    \forall \gamma \in \jump{\Gamma}, p \in \jump{\Gamma \vdash P}(\gamma)
  \end{equation*}
  holds.
\end{corollary}

\section{Application}\label{application}
Let's compare Werner's classical model with our intuitionistic model on some simple cases.
\subsection{Classical model}
We start with the simplest case.
Let the topological space be the simplest one, which is the trivial topological space with its base set the singleton $\{ \phi \}$.
\begin{eqnarray*}
X &:=& 1 = \{\phi\}, \\
\mathcal{O}(X) &:=& \{0,1\} = \{\phi, \{\phi\}\}, \\
p &:=& 0 = \phi
\end{eqnarray*}
This coincides with Werner's Model~\cite{SetsInTypes}.
However this model is so coarse that it represents classical logic, since the principle of excluded middle holds.
\begin{eqnarray*}
  0 &\in& \jump{\forall P : \mathrm{Prop}. P \lor \neg P} \\
  &=& \bigsqcap_{o \in \mathcal{O}(X)} o \lor \neg o = 1.
\end{eqnarray*}

If we want to be more discriminating, we need more opens set in $\mathcal{O}(X)$.

\subsection{Models disproving excluded middle}

Now, let us consider the next simplest topological space.
To do this, we add a new point `$1$' and a new open set $\{\phi, \{\phi\}\}$ into the topological space.
\begin{eqnarray*}
X &:=& 2 = \{0,1\}, \\
\mathcal{O}(X) &:=& \{0,1,2\} = \{\phi, \{\phi\}, \{\phi, \{\phi\}\}\}, \\
p &:=& 1 = \{ \phi \}.
\end{eqnarray*}
Although this model stays simple, its topological space is fine enough to avoid the principle of excluded middle, since the following statement holds.
\begin{equation*}
  1 \notin \jump{\forall P : \mathrm{Prop}. P \lor \neg P} = 1.
\end{equation*}
This statement is derived by using the following equations.
\begin{eqnarray*}
  \neg 0 &=& 2, \\
  \neg 1 &=& 0, \\
  \neg 2 &=& 0.
\end{eqnarray*}
By our soundness theorem, this proves that the principle of excluded middle cannot be deduced in ECC.

\begin{table}[t]
  \begin{center}
    \begin{tabular}{|c||c|c|c|} \hline
      $x^y$ & $0$ & $1$ & $2$ \\ \hline \hline
      $0$ & $2$ & $0$ & $0$ \\ \hline
      $1$ & $2$ & $2$ & $1$ \\ \hline
      $2$ & $2$ & $2$ & $2$ \\ \hline
    \end{tabular}
    \caption{value of $x^y$ to disprove the principle of excluded middle}
    \label{tab:valuea}
  \end{center}
\end{table}



Yet this model is not fully intutionistic as the linearity axiom $(P \rightarrow Q) \lor (Q \rightarrow P)$ holds,
since we have the following fact by Table \ref{tab:valuea}.

\begin{eqnarray*}
  && \jump{\forall P : \mathrm{Prop}. \forall Q : \mathrm{Prop}. (P \rightarrow Q) \lor (Q \rightarrow P)} \\
  &=& \bigsqcap_{o_1, o_2 \in \mathcal{O}(X)} o_1^{o_2} \lor o_2^{o_1} \\
  &=& 2.
\end{eqnarray*}
This is actually interesting because it shows that we can use this model to prove non trivial facts,
for instance that the excluded middle cannot be deduced from the linearity axiom in ECC.
Indeed,
\begin{equation*}
  \jump{(\forall P : \mathrm{Prop}. \forall Q : \mathrm{Prop}. (P \rightarrow Q) \lor (Q \rightarrow P)) \quad \rightarrow \quad (\forall P : \mathrm{Prop}. P \lor \neg P) } = 1.
\end{equation*}

By our soundness theorem, this equation means that there is no term proving the above implication in ECC.\par

\begin{table}[t]
  \begin{center}
    \begin{tabular}{|c||c|c|c|c|c|} \hline
      $x^y$ & $\phi$ & $\alpha$ & $\beta$ & $\gamma$ & $X$ \\ \hline \hline
      $\phi$ & $X$ & $\phi$ & $\phi$ & $\phi$ & $\phi$ \\ \hline
      $\alpha$ & $X$ & $X$ & $\alpha$ & $\alpha$ & $\alpha$ \\ \hline
      $\beta$ & $X$ & $\beta$ & $X$ & $\beta$ & $\beta$ \\ \hline
      $\gamma$ & $X$ & $X$ & $X$ & $X$ & $\gamma$ \\ \hline
      $X$ & $X$ & $X$ & $X$ & $X$ & $X$ \\ \hline
    \end{tabular}
    \caption{value of $x^y$ to disprove the linearity axiom}
    \label{tab_valueb}
  \end{center}
\end{table}

By adding more elements we can refine the model further.
Let
\begin{eqnarray*}
X &:=& \{a,b,x\} \\
\mathcal{O}(X) &:=& \{\phi, \alpha, \beta, \gamma, X\}, \\
&=& \{\phi, \{a\}, \{b\}, \{a,b\}, \{a,b,x\}\}, \\
p &:=& x.
\end{eqnarray*}
In this model, $P \rightarrow Q \lor Q \rightarrow P$ does not hold,
since we have the following fact by Table \ref{tab_valueb}.
\begin{equation*}
  x \notin \jump{\forall P : \mathrm{Prop}. \forall Q : \mathrm{Prop}. P \rightarrow Q \lor Q \rightarrow P} = \gamma
\end{equation*}

\section{Reynolds' Paradox}\label{reynolds}
There is a problem when expanding the set theoretical model, which is called Reynolds' paradox~\cite{model_not_set}.
Basically Reynolds' paradox says that if the interpretation of an impredicative sort has more than one element, it causes a cardinality paradox in the set theoretical model.
This seems to be in contradiction with our model, so in this section we will analyze its assumptions.
\subsection{Outline of the Paradox}
Let $\mathbb{T}$ be an impredicative sort, i.e. if $\Gamma \vdash A : s$ and $\Gamma ; (x : A) \vdash B : \mathbb{T}$ are derivable for any sort $s$ then $\Gamma \vdash \forall x : A.B : \mathbb{T}$ is derivable.
We assume that there exists a type $B$ whose sort is $\mathbb{T}$ such that $\jump{B}$ has at least two elements, i.e.
\begin{equation*}
  \vdash B : \mathbb{T} \quad and \quad \sharp \jump{B} \geq 2.
\end{equation*}


In~\cite{model_not_set} Reynolds says that the existence of such a term $B$ causes a paradox in set-theoretical models.
First, we define the category $\mathbf{Sets}_\mathbb{I}$ and the endofunctor $T$ of $\mathbf{Sets}_\mathbb{I}$.

\begin{definition}
  \
  \begin{itemize}
  \item Let $\mathbf{Sets}_\mathbb{I}$ be a category with:
    \begin{itemize}
    \item $\mathrm{Obj}(\mathbf{Sets}_\mathbb{I}) := \{ \jump{P} \; | \; \vdash P : \mathbb{I}$ is derivable $\}$
    \item $\mathrm{Hom}(\jump{P_1}, \jump{P_2}) := \jump{P_1} \rightarrow  \jump{P_2} = \{f \; | \; f \; \mathrm{is \; a \; function\; from} \; \jump{P_1} \; \mathrm{to} \; \jump{P_2} \}$
    \end{itemize}
  \item Let $T$ be a endofunctor of $\mathbf{Sets}_\mathbb{I}$ with
    \begin{itemize}
    \item $T(\jump{P}) := (\jump{P} \rightarrow \jump{B}) \rightarrow \jump{B}$
    \item $T(\rho) := h \in T(\jump{P_1}) \mapsto \{(g, h(g \circ \rho)) | g \in \jump{P_2} \rightarrow \jump{B}\}$
      \\ where $\rho \in \jump{P_1} \rightarrow \jump{P_2}$
    \end{itemize}
  \end{itemize}
\end{definition}

The paper~\cite{model_not_set} claims the following lemma:
\begin{lemma}\label{paradox}
  \
  \begin{itemize}
  \item $\exists u \in \mathrm{Obj}(\mathbf{Sets}_\mathbb{I}), \exists H \in \mathrm{Hom}(Tu,u)$ s.t. \\
    $\quad \forall s \in \mathrm{Obj}(\mathbf{Sets}_\mathbb{I}), \forall f \in \mathrm{Hom}(Ts, s), \exists ! \rho \in \mathrm{Hom}(u,s)$ s.t. \\
    $\quad\quad$ following diagram commutes.
    $$
    \begin{CD}
      Tu @> T\rho >> Ts \\
      @VHVV @VVfV \\
      u @> \rho >> s
    \end{CD}
    $$
  \item $Tu$ and $u$ are equivalent, i.e. $Tu \cong u$.
  \end{itemize}
\end{lemma}
By definition of endofunctor $T$, $\sharp \jump{B} \geq 2$ implies $Tu$ and $u$ have different cardinalities in spite of $Tu$ and $u$ being isomorphism.
Therefore, the existence of a type $B$ of impredicative sort such that $\sharp\jump{B} \geq 2$ causes a paradox.

\subsection{Avoiding the Paradox}
In ECC, we have an impredicative sort $\mathrm{Prop}$,
and there is a type $\mathrm{B}$ of $\mathrm{Prop}$ such that $\sharp\jump{B} \geq 2$.
However, this doesn't cause a paradox.
In fact, to prove the existence of a function $H \in Tu \rightarrow u$, Reynolds constructs a term $t$ of the type $((P \rightarrow B) \rightarrow B) \rightarrow P$ in the proof of lemma 2 in~\cite{model_not_set}, where $P$ is a type such that $\jump{P} = u$.
If $\jump{(P \rightarrow B) \rightarrow B}$ were interpreted as a set theoretical function space, it would cause a paradox in cardinality since $(P \rightarrow B) \rightarrow B \; \cong \; P$ by Lemma \ref{paradox} and $\sharp \jump{B} \geq 2$.
However in our model $\jump{(P \rightarrow B) \rightarrow B}$ is not a function space, i.e. it is not $(\jump{P} \rightarrow \jump{B}) \rightarrow \jump{B}$, but just some open set of $(X, \mathcal{O}(X))$:
\begin{equation*}
  \jump{(P \rightarrow B) \rightarrow B} = \jump{B}^{\jump{B}^{\jump{P}}} \in \mathcal{O}(X)
\end{equation*}
since both $P$ and $B$ are propositional terms.
Thus this discussion moves to the Heyting algebra part of the model where we need not fear such paradox.

\section{Future Work}\label{future}
There are still three remaining questions we would like to answer in the future: whether the $point \; condition$ is really needed to prove soundness, whether we can handle full ECC, without our restrictions on the type system, and how close to completeness is our model.

The point condition is very restrictive.
It seems to allow only discrete models.
Hence we would like to remove it to allow a wider variety of models.
In fact we have not found any counterexample when removing the $point \; condition$, up to now.
\par
We would also like to lift the restrictions on the PI-Type rule, which prohibits statements about proofs, and on the subtyping rule.
They come from the fact that, in the interpretation of contexts, we use the strict interpretation, which restricts all propositional terms to either $\phi$ or the singleton $\{p\}$, so that we cannot build an element when the non-strict interpretation, while being non-empty, does not contain $p$.
We are considering several approaches to overcome this problem.
\par
While this model rejects the excluded middle, it still admits proof-irrelevance
\begin{equation*}
  \forall t_1, t_2, (t_1, t_2 \; is \; proof \; term \; for \; \Gamma) \Rightarrow \jump{\Gamma \vdash t_1}(\gamma) = \jump{\Gamma \vdash t_2}(\gamma).
\end{equation*}
Since the existence of $t$ such that following condition
\begin{equation*}
  \Gamma ; (p_1 : P) ; (p_2 : P) \vdash t : p_1 =_P p_2 \quad (where \; \Gamma \vdash P : \mathrm{Prop} \; is \; derivable)
\end{equation*}
holds is not provable in general, this means that our model does not still reach completeness.
We are now investigating how close to completeness it is.

\appendix
\section{Proof of Soundness}
\begin{proof}[Proof of 2 of Theorem \ref{soundness}]\label{soundness_ind}
  We assume that $p$ is a reference point.
  \begin{enumerate}
  \item Case of Axiom \\
    $\jump{\Gamma \vdash \mathrm{Prop}}(\gamma) \in \jump{\Gamma \vdash \mathrm{Type}_i}(\gamma)$ is clear. Similarly, $\jump{\Gamma \vdash \mathrm{Type}_i}(\gamma) \in \jump{\Gamma \vdash \mathrm{Type}_{i+1}}(\gamma)$ is also clear.
    
  \item Case of Subtyping \\
    The fact that $\jump{\Gamma \vdash A}(\gamma) \in \jump{\Gamma \vdash \mathrm{Type}_i}(\gamma)$ implies $\jump{\Gamma \vdash A}(\gamma) \in \jump{\Gamma \vdash \mathrm{Type}_{i+1}}(\gamma)$ is clear. 


  \item Case of PI-Type \\
    We will show the fact that
    \begin{eqnarray*}
      && \bigl(\forall \gamma, \alpha, \jump{\Gamma \vdash A}(\gamma) \in \jump{\Gamma \vdash s_1}(\gamma) \\
      && \quad \quad and \; \jump{\Gamma ; (x : A) \vdash B}(\gamma, \alpha) \in \jump{\Gamma ; (x : A) \vdash s_2}(\gamma, \alpha)) \bigr) \\
      && \Rightarrow \; \forall \gamma, \jump{\Gamma \vdash \forall x:A.B}(\gamma) \in \jump{\Gamma \vdash s_3}(\gamma).
    \end{eqnarray*}
    There are three cases as follows.
    \begin{itemize}
    \item $PT_{\Gamma, x}(A,B) = T$\\
      By definition of the interpretation of judgment, the following equation
      \begin{equation*}
        \jump{\Gamma \vdash \forall x : A.B}(\gamma) = \prod_{\alpha \in \jump{\Gamma \vdash A}'(\gamma)}\jump{\Gamma ; (x : A) \vdash B}(\gamma, \alpha)
      \end{equation*}
      holds.
      There are the following two cases:
      \begin{itemize}
      \item $A$ is not a propositional term for $\Gamma$ \\
        Since $\jump{\Gamma \vdash A}(\gamma) \in \mathscr{U}(i)$ , $\jump{\Gamma ; (x : A) \vdash B}(\gamma, \alpha) \in \mathscr{U}(j)$ for any $\gamma, \alpha$ and Lemma \ref{prod_univ}, we have
        \begin{equation*}
          \displaystyle\prod_{\alpha \in \jump{\Gamma \vdash A}(\gamma)} \jump{\Gamma ; (x : A) \vdash B}(\gamma, \alpha) \in \mathscr{U}(\max(i,j)).
        \end{equation*}
      \item $A$ is a propositional term for $\Gamma$ \\
        Since $\jump{\Gamma \vdash A}'(\gamma) \in \mathcal{U}(j)$ , $\jump{\Gamma ; (x : A) \vdash B}(\gamma, \alpha) \in \mathscr{U}(j)$ for any $\gamma, \alpha$ and Lemma \ref{prod_univ}, we have
        \begin{equation*}
          \displaystyle\prod_{\alpha \in \jump{\Gamma \vdash A}'(\gamma)} \jump{\Gamma ; (x : A) \vdash B}(\gamma, \alpha) \in \mathscr{U}(j).
        \end{equation*}
      \end{itemize}
      Hence, the statement holds.

    \item $PT_{\Gamma, x}(A,B) = \TP$ \\
      It is clear since $\jump{\Gamma \vdash \forall x : A.B}(\gamma)$ is an open set by definition of the interpretation of judgment.
    \item $PT_{\Gamma, x}(A,B) = \PP$ \\
      It is clear since $\jump{\Gamma \vdash \forall x : A.B}(\gamma)$ is an open set by definition of the interpretation of judgment.

    \end{itemize}

  \item Case of Abstraction \\
    We will show the fact that
    \begin{eqnarray*}
      && \bigl(\forall \gamma, \alpha, \jump{\Gamma ; (x : A) \vdash t}(\gamma, \alpha) \in \jump{\Gamma ; (x : A) \vdash B}(\gamma, \alpha) \\
      && \quad and \; \jump{\Gamma \vdash \forall x : A.B}(\gamma) \in \jump{\Gamma \vdash s}(\gamma) \bigr) \\
      && \Rightarrow \; \forall \gamma, \jump{\Gamma \vdash \lambda x : A.t}(\gamma) \in \jump{\Gamma \vdash \forall x : A.B}(\gamma).
    \end{eqnarray*}
    There are three cases as follows.
    \begin{itemize}
    \item $PT_{\Gamma, x}(A,B)=T$ \\
      By definition of the interpretation, we have the following equations:
      \begin{eqnarray*}
        \jump{\Gamma \vdash \lambda x : A.t}(\gamma) &=& \Bigl\{\bigl(\alpha, \jump{\Gamma ; (x : A) \vdash t}(\gamma, \alpha) \bigr) \; | \; \alpha \in \jump{\Gamma \vdash A}'(\gamma) \Bigr\}, \\
        \jump{\Gamma \vdash \forall x : A.B}(\gamma) &=& \prod_{\alpha \in \jump{\Gamma \vdash A}'(\gamma)}\jump{\Gamma ; (x : A) \vdash B}(\gamma, \alpha). 
      \end{eqnarray*}
      Then, we must prove the following equation:
      \begin{eqnarray*}
        \Bigl\{\bigl(\alpha, \jump{\Gamma ; (x : A) \vdash t}(\gamma, \alpha) \bigr) \; | \; \alpha \in \jump{\Gamma \vdash A}'(\gamma) \Bigr\} \in \prod_{\alpha \in \jump{\Gamma \vdash A}'(\gamma)}\jump{\Gamma ; (x : A) \vdash B}(\gamma, \alpha).
      \end{eqnarray*}
      But it is clear\footnote{
        If $\jump{\Gamma \vdash A}'(\gamma)$ is the empty set, then $\jump{\Gamma \vdash \forall x :A.B}(\gamma) = \{ \phi \}$ and $\jump{\Gamma \vdash \lambda x : A.t}(\gamma) = \phi$ hold.} by induction of hypothesis.
    \item $PT_{\Gamma, x}(A,B) = \TP$ \\
      Since $\lambda x : A.t$ is a proof term, we have following equations
      \begin{equation*}
        \jump{\Gamma \vdash \lambda x : A.t}(\gamma) = p.
      \end{equation*}
      Hence, the fact we must prove is 
      \begin{equation*}
        p \in \jump{\Gamma \vdash \forall x : A.B}(\gamma).
      \end{equation*}
      By definition we have the following equation.
      \begin{equation*}
        \jump{\Gamma \vdash \forall x : A.B}(\gamma) = \bigsqcap\{ \jump{\Gamma ; (x : A) \vdash B}(\gamma, \alpha) \; | \; \alpha \in \jump{\Gamma \vdash A}(\gamma) \}.
      \end{equation*}
      If $\jump{\Gamma \vdash A}(\gamma)$ is the empty set, then the statement holds since $\jump{\Gamma \vdash \forall x : A.B}(\gamma) = X$ by Lemma~\ref{heyting_conditions}~(\ref{eq:whole}).
      We assume that $\jump{\Gamma \vdash A}(\gamma)$ is a non-empty set.
      We have
      \begin{equation*}
        \forall \alpha \in \jump{\Gamma \vdash A}(\gamma), p \in \jump{\Gamma ; (x : A) \vdash B}(\gamma, \alpha).
      \end{equation*}
      since $\jump{\Gamma ; (x : A) \vdash t}(\gamma, \alpha) = p$.
      Therefore, we have the following equation:
      \begin{equation*}
        p \in \bigcap\{ \jump{\Gamma ; (x : A) \vdash B}(\gamma, \alpha) \; | \; \alpha \in \jump{\Gamma \vdash A}(\gamma) \}.
      \end{equation*}
      However $\bigsqcap S \neq \bigcap S$ hold in general,
      since $\bigsqcap S$ is the interior of $\bigcap S$ when $S$ is non empty subset of $X$.
      Now, we apply the point condition here\footnote{
        This is the place we need it in the proof.
      }.
        We have
      \begin{eqnarray*}
        \jump{\Gamma \vdash \forall x : A. B}(\gamma) &=& \bigsqcap \{ \jump{\Gamma ; (x : A) \vdash B}(\gamma, \alpha) \; | \; \alpha \in \jump{\Gamma \vdash A}(\gamma) \} \\
        &=& \bigcap \{ \jump{\Gamma ; (x : A) \vdash B}(\gamma, \alpha) \; | \; \alpha \in \jump{\Gamma \vdash A}(\gamma) \}
      \end{eqnarray*}
      since $\bigcap\{ \jump{\Gamma ; (x : A) \vdash B}(\gamma, \alpha) \; | \; \alpha \in \jump{\Gamma \vdash A}(\gamma) \}$ is an open set by the point condition.
      Hence, the condition holds in this case.

    \item $PT_{\Gamma, x}(A,B) = \PP$ \\
      Since $\lambda x : A.B$ is a proof term, we have the following equation
      \begin{equation*}
        \jump{\Gamma \vdash \lambda x : A.t}(\gamma) = p.
      \end{equation*}
      Hence, the fact we must prove is 
      \begin{equation*}
        p \in \jump{\Gamma \vdash \forall x : A.B}(\gamma)
      \end{equation*}
      By definition of the interpretation of judgment, we have
      \begin{equation*}
        \jump{\Gamma \vdash \forall x : A.B}(\gamma) = \biggl(\jump{\Gamma \vdash B}(\gamma)\biggr)^{\jump{\Gamma \vdash A}(\gamma)}.
      \end{equation*}
      By characteristic of Heyting algebra,
      \begin{equation*}
        \jump{\Gamma \vdash B}(\gamma) \subset \jump{\Gamma \vdash \forall x : A.B}(\gamma).
      \end{equation*}
      By induction hypothesis $p \in \jump{\Gamma \vdash B}(\gamma)$, so that the condition holds in this case.
    \end{itemize}

  \item Case of Apply \\
    We will show the fact that
    \begin{eqnarray*}
      && \bigl(\forall \gamma, \jump{\Gamma \vdash u}(\gamma) \in \jump{\Gamma \vdash \forall x : A.B}(\gamma) \; and \; \jump{\Gamma \vdash v}(\gamma) \in \jump{\Gamma \vdash A}(\gamma) \bigr) \\
      && \Rightarrow \; \forall \gamma, \jump{\Gamma \vdash u \; v}(\gamma) \in \jump{\Gamma \vdash B[x \backslash v]}(\gamma).
    \end{eqnarray*}
    There are three cases as follows.
    \begin{itemize}
    \item $PT_{\Gamma, x}(A,B) = \T$ \\
      By definition of the interpretation of judgment, the following equation
      \begin{eqnarray*}
        \jump{\Gamma \vdash u \; v}(\gamma) &=& \jump{\Gamma \vdash u}(\gamma) \bigl(\jump{\Gamma \vdash v}(\gamma) \bigr) \\
        \jump{\Gamma \vdash u}(\gamma) &\in& \prod_{\alpha \in \jump{\Gamma \vdash A}'(\gamma)} \jump{\Gamma ; (x : A) \vdash B}(\gamma, \alpha)
      \end{eqnarray*}
      holds. Therefore, we have
      \begin{equation*}
        \jump{\Gamma \vdash u \; v}(\gamma) \in \jump{\Gamma ; (x : A) \vdash B}(\gamma, \jump{\Gamma \vdash v}(\gamma)).
      \end{equation*}
      By Lemma \ref{substitution_interpretation}, we have
      \begin{equation*}
        \jump{\Gamma ; (x : A) \vdash B}(\gamma, \jump{\Gamma \vdash v}(\gamma)) = \jump{\Gamma \vdash B[x \backslash v]}(\gamma).
      \end{equation*}
      Hence, the statement holds in this case.

    \item $PT_{\Gamma, x}(A,B) = \TP$ \\
      It suffices to show that $p \in \jump{\Gamma \vdash B[x \backslash v]}(\gamma)$, since $\jump{\Gamma \vdash u}(\gamma) = \jump{\Gamma \vdash u \; v}(\gamma) = p$ holds.
      By induction hypothesis, we have the following equation
      \begin{equation*}
        p \in \bigsqcap\{ \jump{\Gamma ; (x : A) \vdash B}(\gamma, \alpha) \; | \; \alpha \in \jump{\Gamma \vdash A}(\gamma) \}.
      \end{equation*}
      This equation implies the fact that
      \begin{equation*}
        \forall \alpha \in \jump{\Gamma \vdash A}(\gamma), p \in \jump{\Gamma ; (x : A) \vdash B}(\gamma, \alpha).
      \end{equation*}
      By Lemma \ref{substitution_interpretation} and the fact $\jump{\Gamma \vdash v}(\gamma) \in \jump{\Gamma \vdash A}(\gamma)$, we have
      \begin{equation*}
        p \in \jump{\Gamma \vdash B [x \backslash v]}(\gamma).
      \end{equation*}
      Hence, the statement holds in this case.

    \item $PT_{\Gamma, x}(A,B) = \PP$ \\
      It suffices to show that $p \in \jump{\Gamma \vdash B}(\gamma)$, since $\jump{\Gamma \vdash u}(\gamma) = \jump{\Gamma \vdash v}(\gamma) = \jump{\Gamma \vdash u \; v}(\gamma) = p$ holds and the variable $x$ does not appear freely in $B$.
      The following equation holds.
      \begin{equation*}
        \jump{\Gamma \vdash \forall x : A.B}(\gamma) = \biggl(\jump{\Gamma \vdash B}(\gamma) \biggr)^{\jump{\Gamma \vdash A}(\gamma)}    
      \end{equation*}
      By definition of Heyting algebra, we have
      \begin{equation*}
        \jump{\Gamma \vdash \forall x : A.B}(\gamma) \cap \jump{\Gamma \vdash A}(\gamma) \subset \jump{\Gamma \vdash B}(\gamma).
      \end{equation*}
      Then we have
      \begin{equation*}
        p \in \jump{\Gamma \vdash B}(\gamma).
      \end{equation*}
      by lemma \ref{substitution_interpretation}.
      Hence, the statement holds in this case.
    \end{itemize}

  \item Case of Variable \\
    We must show that
    \begin{eqnarray*}
      &&\bigl( (x : A) \in \Gamma \quad and \quad \forall \gamma, \jump{\Gamma \vdash A}(\gamma) \in \jump{\Gamma \vdash s}(\gamma) \bigr) \\
      && \Rightarrow \forall \gamma, \jump{\Gamma \vdash x}(\gamma) \in \jump{\Gamma \vdash A}(\gamma).
    \end{eqnarray*}
    It is clear by definition of $\jump{\Gamma}$.

  \item Case of Beta Equality \\
    We must show that
    \begin{eqnarray*}
      && \bigl(\forall \gamma, \jump{\Gamma \vdash x}(\gamma) \in \jump{\Gamma \vdash A}(\gamma) \; and \; A =_\beta B \bigr) \\
      && \Rightarrow \forall \gamma, \jump{\Gamma \vdash x}(\gamma) \in \jump{\Gamma \vdash B}(\gamma).
    \end{eqnarray*}
    It is clear by Theorem\ref{soundness} (1).
  \end{enumerate}
\end{proof}






\begin{thebibliography}{50}
\bibitem{pure_type_system} Henk Barendregt. \textsl{Introduction to generalized type systems}. Journal of Functional Programming, 1(2):125–154, 1991.
\bibitem{lambda_type} Henk Barendregt, Wil Dekkers, and Richard Statman. \textsl{Lambda calculus with types}. Cambridge University Press, 2013.
\bibitem{CC} Thierry Coquand and Gerard Huet. \textsl{The calculus of constructions}. Information and computation, 76(2):95–120, 1988.
\bibitem{CategoryType} Bart Jacobs. \textsl{Categorical Logic and Type Theory}. Study in Logic and the Foundationss of Mathemtics 141. Elsevier, 2001.
\bibitem{Luo} Zhaohui Luo. \textsl{A higher-order calculus and theory abstraction. Information and Computation}, 90(1):107–137, 1991.
\bibitem{MacLane} Saunders MacLane and Ieke Moerdijk. \textsl{Sheaves in geometry and logic: A first introduction to topos theory}. Springer, 1992.
\bibitem{coherent} Alexandre Miquel. \textsl{A Model for Impredicative Type Systems, Universe, Intersection Types, and Subtyping}. Proceedings the 15th Annual IEEE Symposium on Logic in Computer Science:18-29, 2000.
\bibitem{not_simple} Alexandre Miquel and Benjamin Werner. \textsl{The not so simple proof-irrelevant model of CC}. In {\it Types for Proof and Programs}. Vol.2646 of Lecture Notes in Computer Science. 240–258, 2003.
\bibitem{model_not_set} John Reynolds. \textsl{Polymorphism is not set-theoretic}. In {\it Semantics of Data Types}. Vol.173 of Lecture Notes in Computer Science. 145–156, 1984.
\bibitem{IntuitionisticLogic} Dirk van Dalen. \textsl{Intuitionistic logic}. Handbook of Philosophical Logic, III:225–339, 1984.
\bibitem{SetsInTypes} Benjamin Werner. \textsl{Sets in types, types in sets}. In {Theoretical aspects of computer software}. Vol.1281 of Lecture Notes in Computer Science. 530–546, 1997.
\bibitem{HoTT} Univalent Foundations Program. \textsl{Homotopy Type Theory: Univalent Foundations of Mathematics.} \url{http://homotopytypetheory.org/book}, 2013.
\end{thebibliography}


\end{document}